\documentclass{article}      
\usepackage{graphicx}
\usepackage{xcolor}

\usepackage{amssymb}
\usepackage{eucal}
\usepackage{amsmath}
\usepackage{amsthm}
\usepackage{graphicx}
\usepackage{float}
\usepackage{framed}
\usepackage{enumerate}
\usepackage{hyperref}
\usepackage[toc,page]{appendix}

\newcommand{\dg}{\dagger}
\newcommand{\UN}{\mathcal{U}(N)}
\newcommand{\SUN}{\mathcal{SU}(N)}
\newcommand{\R}{{\mathbb  R}}  \numberwithin{equation}{section} \newtheorem{thm}{\bf Theorem}[section]
\newtheorem{lem}[thm]{\bf Lemma} \newtheorem{prop}[thm]{\bf Proposition}

\DeclareMathOperator{\tr}{tr}

\begin{document}

\title{Constraint optimization and $\SUN$ quantum control landscapes}

\author{Petre Birtea, Ioan Ca\c su, Dan Com\u{a}nescu\\
{\small Department of Mathematics, West University of Timi\c soara} 
\\
{\small Bd. V. P\^ arvan, No 4, 300223 Timi\c soara, Rom\^ania}\\
{\small Email: petre.birtea@e-uvt.ro, ioan.casu@e-uvt.ro, dan.comanescu@e-uvt.ro}}
\date{}

\maketitle

\begin{abstract}
We develop the embedded gradient vector field method, introduced in \cite{birtea-comanescu} and \cite{Birtea-Comanescu-Hessian}, for the case of the special unitary group $\SUN$ regarded as a constraint submanifold of the unitary group $\UN$. The optimization problem associated to the trace fidelity cost function defined on $\SUN$ that appears in the context of $\SUN$ quantum control landscapes is completely solved using the embedded gradient vector field method. We prove that for $N\geq 5$, the landscape is not $\SUN$-trap free, there are always kinematic local extrema that are not global extrema.
\end{abstract}

\maketitle

\section{Introduction}

In the famous paper \cite{Rabitz-Science} it has been asserted that no sub-optimal local extrema exist as traps for the transition probability landscape corresponding to quantum control systems. However, this has been proved to be more intricate than initially thought, as it was first pointed out in \cite{pechen-1} and subsequently by many authors, including \cite{fouquieres}, \cite{zhdanov-1}, \cite{zhdanov-2}, among others. The sub-optimal controls can appear due to the complexity of the optimization problem, as it is an optimization problem for the composition of two maps, the so-called endpoint map and the fidelity function. Up to this day, the problem of local traps in quantum landscapes is not completely clarified. In the present paper our objective is to approach this problem from a different optimization perspective. We present the mathematical model of the real life problem of quantum control, we point out the difficulties appearing in analyzing the quantum landscape, and we propose a new method to tackle this problem. Most studies have been done for $\UN$ quantum control landscapes. For a detailed explanation and physical interpretations, see \cite{rabitz-2007}, \cite{Rabitz-relaxed}, \cite{palao}, \cite{Rabitz-Science}, \cite{russell}, \cite{russell-wu}, and \cite{shlomo}. In this paper we will study the case of $\SUN$ quantum control landscapes. This study involves solving an optimization problem on $\SUN$ which we will show that it is quite different from an analogous optimization problem on $\UN$.

The quantum control behavior on the special unitary group $\SUN$ is described by the Schr\"{o}dinger equation
\begin{equation}\label{control}
\imath\hbar \dot U(t)=\left(H_0+\sum_{i=1}^mu_i(t)H_i\right)U(t),~~~U(0)=U_0\in \SUN,
\end{equation}
where $H_0,H_1,\dots,H_m$ are traceless Hermitian operators, the control vector $\bar u=(u_1,\dots,u_m)$ belongs to  $L^2([0,T],\mathbb{R}^m)$, and $N$ is the number of quantum states. In what follows, we will always assume that the above quantum control system is globally controllable. For a more detailed discussion on this hypothesis see \cite{altafini}, \cite{russell}, and \cite{russell-wu}.

To the differential control system \eqref{control}, we associate the endpoint map based at $U_0$,
$\hbox{End}_{U_0}:L^2([0,T],\mathbb{R}^m)\to \SUN,~\hbox{End}_{U_0}(\bar u)=U_{\bar u}(T),$
where $U_{\bar u}(T)$ is the point reached at time $T$ going along the solution of the system \eqref{control} starting at the initial point $U_0$. The endpoint map is smooth in the Fr\' echet sense and a formula for its differential is explicitly given in \cite{agrachev} and \cite{montgomery}. 

The performance function used in quantum control theory, called \textbf{quantum control landscape}, is
$J:L^2([0,T],\mathbb{R}^m)\to \mathbb{R},~~J(\bar u)=\dfrac{1}{N}\hbox{Re tr}\left(A^{\dagger}U_{\bar{u}}(T)\right),$
where $A$ is a desirable quantum target gate.

The quantum control system \eqref{control} is called $\SUN$-\textbf{trap free} if the function $J$ has only global extrema (initial definition in \cite{rabitz-2007}, \cite{Rabitz-Science}) or global extrema and saddle points (relaxed definition in \cite{Rabitz-relaxed}).

The quantum control landscape $J$ is the composition of two maps,
$$J(\bar u)=\widetilde{G}\circ \hbox{End}_{U_0}(\bar u),$$
where $\widetilde{G}:\SUN\to \mathbb{R},~\widetilde{G}(U)=\dfrac{1}{N}\hbox{Re tr}\left(A^{\dagger}U\right)$ is called the \textbf{trace fidelity function}. In applied problems other types of cost functions $\widetilde{G}$ are also considered, see \cite{dominy-rabitz}, \cite{fouquieres}, \cite{pechen-2}, \cite{pechen-3}, \cite{pechen-5}, \cite{pechen-1},  \cite{pechen-4}, \cite{Rabitz-Science}, \cite{russell-wu}.

The critical points of $\widetilde{G}$ are called \textit{kinematic} and those of $J$ are called \textit{dynamic}, see \cite{dominy-rabitz}, \cite{pechen-1}, \cite{russell-wu}. There are two layers in finding and characterizing the critical points of the quantum control landscape functional $J$, due to the fact that the differential of $J$ is a composition of two linear maps, as follows,
$$
L^2([0,T],\mathbb{R}^m) \xrightarrow{d_{\bar u}\text{End}_{U_0}} T_{_{\text{End}_{U_0}({\bar u})}}\SUN \xrightarrow{d_{\text{End}_{U_0}({\bar u})}\widetilde{G}} \R.$$

{\bf First layer}. Let ${\bar u}\in L^2([0,T],\mathbb{R}^m)$ be a control vector such that $\text{End}_{U_0}({\bar u})$ is a kinematic critical point, i.e. the linear map $d_{\text{End}_{U_0}({\bar u})}\widetilde{G}$ is identically zero. Consequently, ${\bar u}$ is also a dynamical critical point.

{\bf Second layer}.  Let ${\bar u}\in L^2([0,T],\mathbb{R}^m)$ be a control vector such that $\text{End}_{U_0}({\bar u})$ is {\bf NOT} a kinematic critical point, i.e. $\text{End}_{U_0}({\bar u})$ is not a critical point for $\widetilde{G}$. Nevertheless, $\bar{u}$ can still be a dynamical critical point, i.e. a critical point for quantum control landscape $J$.
As $\text{End}_{U_0}({\bar u})$ is not a kinematic critical point, we have $\nabla_{\SUN}\widetilde{G}(\text{End}_{U_0}({\bar u}))\neq 0$ and the following space decomposition holds\footnote{$\left<\nabla_{\SUN}\widetilde{G}(\text{End}_{U_0}({\bar u}))\right>$ is the linear subspace generated by $\nabla_{\SUN}\widetilde{G}(\text{End}_{U_0}({\bar u}))$.}
$$T_{_{\text{End}_{U_0}({\bar u})}}\SUN =\text{Ker} d_{\text{End}_{U_0}({\bar u})}\widetilde{G}\oplus \left<\nabla_{\SUN}\widetilde{G}(\text{End}_{U_0}({\bar u}))\right>.$$
Consequently, a control vector ${\bar u}\in L^2([0,T],\mathbb{R}^m)$, with the property that $\text{End}_{U_0}({\bar u})$ is not a kinematic critical point, is a dynamic critical point if and only if\footnote{The perpendicularity in (ii) is taken with respect to the bi-invariant metric on $\SUN$.} $\nabla_{\SUN}\widetilde{G}(\text{End}_{U_0}({\bar u}))\perp \text{Im}d_{\bar u}\text{End}_{U_0}$. This perpendicularity condition and $\text{End}_{U_0}({\bar u})$  not being a kinematic critical point imply that $d_{\bar u}\text{End}_{U_0}$ is not full rank.
\medskip

The above necessary and sufficient condition can be written equivalently in terms of non-transver\-sality property of the endpoint map as it is presented in \cite{russell-wu}. Although, dynamical critical points can appear, see \cite{pechen-1}, using results from differential topology regarding generic property of transversality condition, see \cite{abraham}, it is argued that this situation is rare (of null measure). It is concluded in \cite{russell-wu} that "dynamical landscape almost always possesses the same critical point structure as the kinematical one". 
The absence of traps in the functional (not in the kinematic) space has been rigorously proved in \cite{pechen-2}, \cite{pechen-3}, \cite{pechen-4}, \cite{pechen-5} in problems like the Landau-Zener system, the control of quantum transmission of particles through a potential barrier, ultra-fast controls. 
\medskip

In this paper, using embedded gradient vector field method, we will prove that for $N\geq 5$, the landscape is not $\SUN$-trap free (there are always kinematic local extrema that are not global extrema).\medskip

\section{Optimization on the special unitary group $\SUN$}

We begin by recalling the construction of the embedded gradient vector field method and we compute this vector field for a cost function defined on the special unitary group $\SUN$, that we regard as a constraint submanifold of the unitary group $\UN$.
We will apply these results in the next section to the trace fidelity optimization problem on $\SUN$. Optimization on constraint manifolds is an old problem in mathematics with a huge amount of applications. One of the most fruitful approach is the projected gradient method, see the recent works \cite{absil-carte}, \cite{absil-malick}, \cite{boumal-2}, \cite{boumal-1}, \cite{edelman}, \cite{gao-1}, \cite{gao-2}, \cite{hu}, \cite{mishra}.

Let $\mathcal{S}\subset \mathfrak{M}$ be a submanifold of a Riemannian manifold $(\mathfrak{M},{\bf g})$, that can be described by a set of constraint functions, i.e. $\mathcal{S}={\bf F}^{-1}(c)$, where ${\bf F}=(F_1,\dots,F_k):\mathfrak{M}\rightarrow \R^k$ is a smooth map and $c\in \R^k$ is a regular value of ${\bf F}$. We endow $\mathcal{S}$ with the induced metric, hence $(\mathcal{S},{\bf g}_{_{ind}})$ becomes itself a Riemannian manifold. 

For solving optimization problems one needs, in general, to compute the gradient vector field and the Hessian operator of a smooth cost function $\widetilde{G}:(\mathcal{S},{\bf g}_{_{ind}})\rightarrow \R$. The Riemannian geometry of the submanifold $\mathcal{S}$ can be more complicated than the Riemannian geometry of the ambient manifold $\mathfrak{M}$. In what follows, we show how we can compute the gradient vector field and the Hessian operator of $\widetilde{G}$ using only the geometry of the ambient manifold $(\mathfrak{M},{\bf g})$.

Let $G:(\mathfrak{M},{\bf g})\rightarrow \R$ be a smooth prolongation of $\widetilde{G}$, i.e. $\widetilde{G}=G_{|\mathcal{S}}$. In \cite{birtea-comanescu}, \cite{Birtea-Comanescu-Hessian}, \cite{5-electron}, it has been proved that
\begin{equation}\label{got}
\nabla _{{\bf g}_{_{ind}}}\widetilde{G}(s)=\partial_{\bf g} G(s),\,\,\,\forall s\in \mathcal{S},
\end{equation}
where $\partial_{\bf g} G$ is the unique vector field defined on the open set of regular points of the constraint function $\mathfrak{M}^{reg}\subset \mathfrak{M}$ that is tangent to the foliation generated by ${\bf F}$ having property \eqref{got}. We call $\partial_{\bf g} G$ {\bf the embedded gradient vector field} and it  is given by the following formula:
\begin{equation}\label{embedded}
\partial_{\bf g} G(s)=\nabla_{\bf g} G(s)-\sum\limits_{i=1}^k\sigma_{\bf g}^{i}(s)\nabla_{\bf g}F_i(s).
\end{equation}
The Lagrange multiplier functions $\sigma_{\bf g}^{i}:\mathfrak{M}^{reg}\rightarrow \R$ are defined by

\begin{equation}\label{sigma-101}
\sigma^i_{\bf g}(s):=\frac{\det \left(\text{Gram}_{(F_1,\ldots ,F_{i-1},G, F_{i+1},\dots,F_k)}^{(F_1,\ldots , F_{i-1},F_i, F_{i+1},...,F_k)}(s)\right)}{\det\left(\text{Gram}_{(F_1,\ldots ,F_k)}^{(F_1,\ldots ,F_k)}(s)\right)},
\end{equation}
where
\begin{equation*}\label{sigma}
\text{Gram}_{(g_1,...,g_s)}^{(f_1,...,f_r)}=\left[%
\begin{array}{cccc}
  {\bf g}(\nabla_{\bf g} g_1,\nabla_{\bf g}f_{1}) & ... & {\bf g}(\nabla_{\bf g} g_s,\nabla_{\bf g} f_{1}) \\
  \vdots & \ddots & \vdots \\
  {\bf g}(\nabla_{\bf g} g_1,\nabla_{\bf g}f_r) & ... & {\bf g}(\nabla_{\bf g} g_s,\nabla_{\bf g} f_r) 
\end{array}%
\right].
\end{equation*}

Also, in \cite{Birtea-Comanescu-Hessian}, \cite{5-electron} it has been proved that 
\begin{equation}\label{Hessian-general}
\text{Hess}_{{\bf g}_{_{ind}}}\, \widetilde{G}(s) =\left(\text{Hess}_{\bf g}\,G(s)-\sum_{i=1}^k\sigma_{\bf g}^{i}(s)\text{Hess}_{\bf g}\, F_i(s)\right)_{|T_s \mathcal{S}\times T_s \mathcal{S}}.
\end{equation}

The submanifold $\mathcal{S}$ will be the constraint submanifold $\SUN$ endowed with the induced metric of the unitary group $\mathcal{U}(N)$, which will play the role of the ambient manifold $\mathfrak{M}$. We solve the optimization problem generated by ${G}_{|{\SUN}}$, where $G:\mathcal{U}(N)\rightarrow \R$ is a smooth cost function. 
More precisely, our optimization problem is
$$
\underset{S\in \SUN}{\mathrm{argmin}}~{G}_{|\SUN}(S).
$$

The real Lie group of unitary matrices\footnote{By $U^{\dg}$ we denote the transpose conjugate of the matrix $U$.} is 
$$\mathcal{U}(N)=\{U\in \mathcal{M}_N(\mathbb{C})\,|\,U^{\dagger}U=UU^{\dg}=\mathbb{I}_N\}.$$

The tangent space to the manifold $\mathcal{U}(N)$ in a point $U\in \mathcal{U}(N)$ is given by
\begin{equation*}
T_U \mathcal{U}(N)=\{\Omega U\,|\,\Omega+\Omega^{\dg}=\mathbb{O}_N\}.
\end{equation*}
We consider the bi-invariant Riemannian metric on the unitary group 
$$\left<\Omega_1U,\Omega_2U\right>_{\mathcal{U}(N)}=\frac{1}{2}\text{Re}\tr(\Omega_1^{\dg}\Omega_2).$$

In order to regard the special unitary group as a constraint manifold, we identify the complex numbers with points in $\R^2$ and the complex numbers having modulus 1 with the points in the submanifold $S^1\subset \R^2$. We consider the smooth function $\det :\mathcal{U}(N)\rightarrow S^1$ and $h:S^1\rightarrow \R$ a local chart around the point $(1,0)\in S^1$.

The special unitary group 
$$\mathcal{SU}(N)=\{S\in\UN\,|\,\det S=1\}$$ is the level set $F_h^{-1}(c)$, where $F_h:\UN\rightarrow \R$, $F_h(U)=(h\circ\det) (U)$ and $c=h(1,0)$. 

\begin{lem}
The value $c$ is a regular value for $F_h$, therefore $\SUN$ is a constraint submanifold of $\UN$.
\end{lem}

\begin{proof}
We have to prove that for any $S\in F_h^{-1}(c)$, the linear map $\text{d}_S F_h:T_S\UN\rightarrow \R$ is surjective. By the definition of the constraint function $F_h$, we have 
\begin{equation*}
\text{d}_S F_h=\text{d}_{(1,0)}h\circ \text{d}_S\det,
\end{equation*}
where $\text{d}_S\det:T_S\UN\rightarrow T_{(1,0)}S^1$ and $\text{d}_{(1,0)}h:T_{(1,0)}S^1\rightarrow \R$.

The set equality
$T_{(1,0)}S^1=\{(1,b)\,|\,b\in \R\}$ holds.
The map $\tau(1,b)=b\imath$ identifies the tangent space $T_{(1,0)}S^1$ with the set of purely imaginary complex numbers.

For a tangent vector $\Omega S\in T_S\UN$, we easily obtain that
$\text{d}_S\det (\Omega S)=\det S\tr(S^{\dg}\Omega S)=\tr \Omega$. To write $\text{d}_S\det (\Omega S)$ as a tangent vector in $T_{(1,0)}S^1$, we use the identification $\tau$ and, consequently, 
\begin{equation*}
 \text{d}_S\det (\Omega S)=(1,-\imath\tr\Omega),
 \end{equation*}
which shows that $ \text{d}_S\det$ is a surjective map. 
Also, because $h$ is a local chart, $\text{d}_{(1,0)}h$ is a bijection and hence we obtain that $\text{d}_S F_h$ is a surjective map.

\end{proof}

For explicit computations, we need to use specific choices for the local chart $h$. One example of such a local chart is the stereographic projection from the West, i.e. from the point $(-1,0)\in S^1$. In this case we have $h_{W}:S^1\backslash \{(-1,0)\}\rightarrow \R$, $h_{W}(x,y)=\dfrac{y}{x+1}$. Denoting $z=x+\imath y $ and taking into account that $|z|=1$ for $z\in S^1$, we can write $h_W(z)=\dfrac{\text{Im}(z)}{\text{Re}(z)+1}=-\imath\dfrac{z-1}{z+1}$ and its derivative is $h'_W(z)=-\dfrac{2\imath}{(z+1)^2}$.

Consequently, the constraint function $F_{h_{W}}:\UN\rightarrow \R$ is given by $F_{h_{W}}(U)=\dfrac{\text{Im}(\det U)}{\text{Re}(\det U)+1}$ and we have $\SUN=F_{h_{W}}^{-1}(0)$. 


\subsection{The embedded gradient vector field on $\SUN$}

By a straightforward computation and using the equality $\text{d}_U\det (\Omega U)=\det U\tr(U^{\dg}\Omega U)=\det U\tr\Omega$ for $U\in\UN$, we obtain
\begin{align}\label{diferentiala-UN-1212}
\text{d}_UF_{h_{W}}(\Omega U)  = h'_{W}(\det U)\text{d}_U\det(\Omega U)= -2\imath\frac{\det U}{(\det U+1)^2} \tr \Omega.
\end{align}

For $S\in \SUN$, the above equality becomes
\begin{equation}\label{diferentiala-specific-123}
\text{d}_SF_{h_{W}}(\Omega S)=-\frac{\imath}{2}\tr \Omega.
\end{equation}

For a smooth cost function $G:\UN\rightarrow \R$, using the formulas \eqref{embedded} and \eqref{sigma-101}, we obtain the formula for the gradient vector field on $\SUN$ of the restricted cost function $G_{|{\SUN}}$ expressed as a vector field in the ambient space $\UN$,
\begin{align*}
\nabla_{\SUN}\left(G_{|{\SUN}}\right)(S) & =\partial_{\UN}G(S)  \nonumber\\
& =\nabla_{\UN}G(S)-\frac{\left<\nabla_{\UN}G(S),\nabla_{\UN}F_{h_{W}}(S)\right>_{\UN}}{\left<\nabla_{\UN}F_{h_{W}}(S),\nabla_{\UN}F_{h_{W}}(S)\right>_{\UN}}\nabla_{\UN}F_{h_{W}}(S).
\end{align*}
In order to detail the above formula, we need to compute $\nabla_{\UN}F_{h_{W}}(S)$. By the definition of the gradient, for $S\in \SUN$ and $\Omega$ a skew-Hermitian matrix, we have
$$\text{d}_SF_{h_{W}}(\Omega S)=\left<\nabla_{\UN}F_{h_{W}}(S),\Omega S\right>_{\UN}.$$ 
Taking into account \eqref{diferentiala-specific-123} the previous equality becomes
$$-\frac{\imath}{2}\tr \Omega=\left<\nabla_{\UN}F_{h_{W}}(S)S^{\dg},\Omega\right>_{\UN},$$
which is equivalent to 
$$\left<-\imath\mathbb{I}_N,\Omega\right>_{\UN}=\left<\nabla_{\UN}F_{h_{W}}(S)S^{\dg},\Omega\right>_{\UN}.$$
It follows that 
$$\nabla_{\UN}F_{h_{W}}(S)=-\imath S.$$
Consequently,
$$\left<\nabla_{\UN}F_{h_{W}}(S),\nabla_{\UN}F_{h_{W}}(S)\right>_{\UN}=\frac{N}{2}$$
and 
$$\left<\nabla_{\UN}G(S),\nabla_{\UN}F_{h_{W}}(S)\right>_{\UN}=\frac{\imath}{2}\tr(S^{\dg}\nabla_{\UN}G(S)).$$

Synthesizing the above computations, we obtain the following formula for the gradient vector field of a cost function defined on the special unitary group $\SUN$ and a necessary and sufficient condition for critical points.

\begin{thm}
Let $G:\UN\rightarrow \R$ be a smooth cost function and $S\in \SUN$. Then:
\begin{enumerate}[(i)]
\item the gradient vector field on $\SUN$ of the restricted cost function $G_{|{\SUN}}$ expressed as a vector field in the ambient space $\UN$ is given by

 \begin{equation}\label{gradient-SUN-linear}
 \nabla_{\SUN}\left(G_{|{\SUN}}\right)(S)=\nabla_{\UN}G(S)-\frac{1}{N}\tr\left(S^{\dg}\nabla_{\UN}G(S)\right)S.
 \end{equation}
 
 \item the matrix $S\in \SUN$ is a critical point for $G_{|{\SUN}}$ if and only if there exists $\alpha\in \mathbb{C}$ such that
 \begin{equation*}
 \nabla_{\UN}G(S)=\alpha S.
 \end{equation*}
 \end{enumerate}
 \end{thm}

The gradient vector field $\nabla_{\UN}G$ on the unitary group $\UN$ can be further written as a vector field in the ambient space $\mathcal{M}_N(\mathbb{C})$, see \cite{abrudan} and \cite{manton}. For the cases of $\mathcal{SU}(2)$ and $\mathcal{SU}(3)$, formulas for the gradient vector fields are given in \cite{fiori-su(2)} and \cite{fiori-su(3)}. 

\subsection{Hessian operator on $\SUN$}

In the case of $\SUN$ viewed as a constraint submanifold of $\UN$, the formula \eqref{Hessian-general} for the Hessian matrix becomes
\begin{equation*}
\text{Hess}_{\SUN}\, \left(G_{|{\SUN}}\right)(S)=\left(\text{Hess}_{\UN}\,G(S)-\sigma(S)\text{Hess}_{\UN}\, F_{h_{W}}(S)\right)_{|T_S \SUN\times T_S \SUN},
\end{equation*}
where $\sigma(S)=\dfrac{1}{N}\tr\left(S^{\dg}\nabla_{\UN}G(S)\right)$.
The tangent space to $\SUN$ in a point $S\in \SUN$ is the set 
$$T_S\SUN=\{\Omega S\,|\,\Omega+\Omega^{\dg}=\mathbb{O}_N,\,\tr(\Omega)=0\}.$$

In order to compute $\text{Hess}_{\UN}\, F_{h_{W}}$, we consider the geodesic curve $\gamma:\R\rightarrow \UN$ given by $\gamma(t)= \exp(t\Omega)U$, which passes trough $U\in \UN$ at $t=0$ with the speed $\Omega$ in the Lie algebra $\mathfrak{u}(N)$ of $\UN$.

Using \eqref{diferentiala-UN-1212}, we have\footnote{$\det(\exp X)=\exp(\tr X)$ for $X\in \mathcal{M}_N(\mathbb{C})$.}
\begin{align*}
\left(F_{h_{W}}\circ\gamma\right)'(t) & =\text{d}_{\gamma(t)}F_{h_{W}}(\gamma'(t))=\text{d}_{\gamma(t)}F_{h_{W}}(\Omega\exp(t\Omega) U)\nonumber\\
& = -2\imath\frac{\det (\exp(t\Omega)U)\tr\Omega}{(\det(\exp(t\Omega)U)+1)^2} = -2\imath\frac{\det (\exp(t\Omega))\det U\tr\Omega}{(\det(\exp(t\Omega))\det U+1)^2}\nonumber \\
&  = -2\imath\frac{\exp(t\tr \Omega) \det U\tr\Omega}{(\exp(t\tr\Omega)\det U+1)^2}
\end{align*}
and 
\begin{equation*}
\left(F_{h_{W}}\circ\gamma\right)''(t)=2\imath\frac{\det U(\tr \Omega)^2\exp(t\tr\Omega)\left(\det U\exp(t\tr\Omega)-1\right)}{\left(\det U\exp(t\tr\Omega)+1\right)^3}.
\end{equation*}
We obtain the quadratic form associated with the Hessian matrix in a point $U\in \UN$ 
\begin{equation*}
\text{Hess}_{\UN}\, F_{h_{W}}(U)(\Omega U,\Omega U) = \left(F_{h_{W}}\circ\gamma\right)''(0)=2\imath\frac{\det U(\tr \Omega)^2\left(\det U-1\right)}{\left(\det U+1\right)^3},\,\,\text{for}\,\,\Omega\in \mathfrak{u}(N).
\end{equation*}
By the polarization law\footnote{For a bilinear form $B:\mathcal{V}\times \mathcal{V}\rightarrow \R$ we have $B(u,v)=\dfrac{1}{2}\left(B(u+v,u+v)-B(u,u)-B(v,v)\right)$, for all $u,v\in \mathcal{V}$.}, 
we can compute  $\text{Hess}_{\UN}\, F_{h_{W}}(U)(\Omega_1 U,\Omega_2 U) $, for $\Omega_1,\Omega_2\in \mathfrak{u}(N)$.

If we take $S\in \SUN$ in the above formula, then $\text{Hess}_{\UN}\, F_{h_{W}}(S)(\Omega S,\Omega S) = 0$, for $\Omega\in \mathfrak{u}(N)$ and, consequently,  $\text{Hess}_{\UN}\, F_{h_{W}} (S)(\Omega_1 S,\Omega_2 S)=0$, for $\Omega_1,\Omega_2\in \mathfrak{u}(N)$.

It follows that
\begin{equation}\label{Helmke-formula11}
\text{Hess}_{\SUN}\, \left(G_{|{\SUN}}\right)(S)=\left(\text{Hess}_{\UN}\,G(S)\right)_{|T_S \SUN\times T_S \SUN},
\end{equation}
which is a particular case of the more general result from Proposition 3.3 in \cite{helmke}.

\section{The trace fidelity problem on $\SUN$}

As previewed in Introduction, a cost function of interest in quantum control landscape theory is the 
so-called trace fidelity cost function. As multiplication with a constant factor does not affect the set of critical points and their nature, we will drop the factor $\dfrac{1}{N}$ from the definition of the trace fidelity function and we work in what follows with the cost function
 $\widetilde{G}:\SUN \rightarrow \R$ given by 
$$\widetilde{G}(S)=\text{Re}\tr(A^{\dg}S),$$
where $A\in \UN$ is a given target matrix. Our objective is to determine its critical points and to study their nature. In the quantum control context these critical points are called kinematic critical points.

By a straightforward computation we notice that this cost function comes from a least square problem, since we have\footnote{We denote by $\|X\|_F=\sqrt{\tr\left(X^{\dg}X\right)}$ the Frobenius norm of the matrix $X$ on $\mathcal{M}_N(\mathbb{C})$.} 
\begin{align}\label{least-square-Frobenius}
\|A-S\|_F^2 & = \tr\left((A^{\dg}-S^{\dg})(A-S)\right)=\tr\left(A^{\dg}A\right)+N-\tr\left(A^{\dg}S+(A^{\dg}S)^{\dg}\right) \nonumber \\
& = 2N-2\text{Re}\tr(A^{\dg}S).
\end{align}
Consequently, the critical points of $\widetilde{G}$ coincide with the critical points of the least square cost function; maxima of $\widetilde{G}$ are minima of the least square cost function, and vice versa.

The trace fidelity cost function $\widetilde{G}$ has the natural extension $G:\UN\rightarrow \R$, $G(U)= \text{Re}\tr(A^{\dg}U)$. In order to obtain $\nabla _{\UN}G$, for $U\in \UN$ and $\Omega$ a skew-Hermitian matrix,  we have the following computations:
\begin{align*}
\left<\nabla_{\UN}G(U),\Omega U\right>_{\UN} & = \text{d}_U G(\Omega U) = \text{Re}\tr\left(\text{d}(A^{\dg}U)\Omega U\right) \\
& = \text{Re}\tr\left(A^{\dg}\Omega U\right) =\text{Re}\tr\left(U A^{\dg}\Omega \right)  \\
& = \text{Re}\tr\left(\left(U A^{\dg}\right)_{S}\Omega \right) +\text{Re}\tr\left(\left(U A^{\dg}\right)_{H}\Omega \right) \\
& = \text{Re}\tr\left(\left(U A^{\dg}\right)_{S}\Omega UU^{\dg}\right) = \text{Re}\tr\left(U^{\dg}\left(U A^{\dg}\right)_{S}\Omega U\right) \\
& = 2 \left<\left(\left(U A^{\dg}\right)_{S}\right)^{\dg}U,\Omega U\right>_{\UN} \\
& = \left<\left( AU^{\dg}-UA^{\dg} \right)U,\Omega U\right>_{\UN},
\end{align*}
where we have denoted by $X_S$, $X_H$ the skew-Hermitian and the Hermitian part of the matrix $X$.

It follows that 
$$\nabla_{\UN}G(U) =\left( AU^{\dg}-UA^{\dg} \right)U.$$
Using \eqref{gradient-SUN-linear}, we obtain the expression of the gradient vector field for the trace fidelity cost function
\begin{equation*}
\nabla_{\SUN}\widetilde{G}(S) = \left( AS^{\dg}-SA^{\dg} -\frac{1}{N}\tr\left(AS^{\dg}-SA^{\dg}\right)\mathbb{I}_N    \right)S.
\end{equation*}
Consequently, the critical points of $\widetilde{G}$ are the solutions of the matrix equation
\begin{equation*}
AS^{\dg}-SA^{\dg} =\frac{1}{N}\tr\left(AS^{\dg}-SA^{\dg}\right)\mathbb{I}_N.
\end{equation*}

We notice that the following set equality holds:
$$\left\{X\in \mathcal{M}_N(\mathbb{C})\,\left|\,X=\frac{\tr X}{N}\mathbb{I}_N,\,X\,\,\text{is skew-Hermitian}\right.\right\}=
\left\{\mu\imath \mathbb{I}_N\,\left|\,\mu\in \R\right.\right\}.$$

Synthesizing the above considerations, we obtain the following result.

\begin{prop}\label{critical-condition-fidelity}
The critical points  of the trace fidelity cost function $\widetilde{G}$ are the solutions $S\in \SUN$ of the matrix equations
$$AS^{\dg}-SA^{\dg} =\mu\imath  \mathbb{I}_N,\,\,\,\mu\in \R.$$
\end{prop}

\subsection{Kinematic critical points when the target matrix $A$ is from $\SUN$}

From \eqref{least-square-Frobenius} it follows that the least square cost problem has the matrix $A$ as the unique global minimum and, consequently, {$A$ is the unique global maximum for the trace fidelity cost function $\widetilde{G}$.}

We can reduce the study of $\widetilde{G}$, with an arbitrary $A\in \SUN$, to the particular case of the cost function $\widehat{G}(S)=\text{Re}\tr(S)$. We have the obvious equality $\widetilde{G}(AS)=\widehat{G}(S)$ for any $S\in \SUN$. 
{\bf A matrix $S\in \SUN$ is a critical point for $\widehat{G}$ if and only if the matrix $AS\in \SUN$ is a critical point for $\widetilde{G}$. Moreover, if $S$ is a local minimum, a local maximum, a saddle point for the reduced cost function $\widehat{G}$, then $AS$ is a local minimum, a local maximum, a saddle point, respectively, for the cost function $\widetilde{G}$.}\bigskip

Using Proposition \ref{critical-condition-fidelity}, the equations that describe the critical points of the reduced cost function $\widehat{G}$ are
\begin{equation}\label{critice-I-N}
\begin{cases}
S^{\dg}-S =\mu\imath  \mathbb{I}_N,\,\,\,\mu\in \R \\
S^{\dg}S=SS^{\dg}=\mathbb{I}_N \\
\det(S)=1.
\end{cases}
\end{equation}
Let $S$ be a critical point of the reduced cost function $\widehat{G}$, i.e. a matrix which satisfies \eqref{critice-I-N}. We make the following substitution 
$$W=S+\frac{1}{2}\mu \imath \mathbb{I}_N.$$
The condition $S^{\dg}-S =\mu\imath  \mathbb{I}_N$ becomes $W^{\dg}=W$, i.e. $W$ has to be a Hermitian matrix.
By the spectral decomposition theorem for Hermitian matrices, there exist a unitary matrix $V$ and a diagonal matrix $D=\text{diag}(\lambda_1,\lambda_2,\dots,\lambda_N)$ such that $W=V^{\dg}DV$, where $\lambda_1\geq \lambda_2\geq \dots\geq \lambda_N$ are the real eigenvalues of the Hermitian matrix $W$.

Substituting $S=V^{\dg}DV-\dfrac{1}{2}\mu \imath \mathbb{I}_N$ in the second condition from \eqref{critice-I-N} we obtain 
$$D^2=\left(1-\frac{\mu^2}{4}\right)\mathbb{I}_N.$$
From the above condition it follows that we must have $\mu\in [-2,2]$. For every $k\in \{1,\dots,N\}$, we have $\lambda_k^2=1-\dfrac{\mu^2}{4}$, therefore $\lambda_k=\sqrt{1-\dfrac{\mu^2}{4}}\,\varepsilon_k$, $\varepsilon_k\in\{-1,1\}$.

Imposing the last condition from  \eqref{critice-I-N}, we obtain 
$$\det\left(D-\frac{1}{2}\mu\imath \mathbb{I}_N\right)=1,$$
which is equivalent with 
\begin{equation}\label{produs-b-b}
\prod_{k=1}^N\left(\sqrt{1-\frac{\mu^2}{4}}\,\varepsilon_k-\frac{1}{2}\mu\imath\right)=1.
\end{equation}
We introduce the sets $\mathcal{K}_+=\{k\in\{1,2,\dots,N\}\,|\,\varepsilon_k=1\}$ and its complement $\mathcal{K}_-=\{k\in\{1,2,\dots,N\}\,|$ $\varepsilon_k=-1\}$. The diagonal matrix $D$ is given by
$$D=\sqrt{1-\frac{\mu^2}{4}}\,\text{diag}(\underbrace{1,\dots,1}_{|\mathcal{K}_+|\rm\ times},\underbrace{-1,\dots,-1}_{N-|\mathcal{K}_+|\rm\ times}),$$
where $|\mathcal{K}_+|$ is the cardinal of the set $\mathcal{K}_+$.
The equation \eqref{produs-b-b} can be written
$$(-1)^{N-|\mathcal{K}_+|}\left(\sqrt{1-\frac{\mu^2}{4}}-\frac{1}{2}\mu\imath\right)^{|\mathcal{K}_+|}\left(\sqrt{1-\frac{\mu^2}{4}}+\frac{1}{2}\mu\imath\right)^{N-|\mathcal{K}_+|}=1.$$
 Making the notation $z_{\mu}:=\sqrt{1-\dfrac{\mu^2}{4}}+\dfrac{1}{2}\mu\imath$ and taking into account that $|z_{\mu}|=1$, we obtain that $\dfrac{1}{z_{\mu}}=\sqrt{1-\dfrac{\mu^2}{4}}-\dfrac{1}{2}\mu\imath$. 

Consequently, the equation \eqref{produs-b-b} becomes
\begin{equation*}
z_{\mu}^{N-2|\mathcal{K}_+|}=(-1)^{N-|\mathcal{K}_+|},\,\,\,\text{Re}(z_{\mu})\geq 0.
\end{equation*}

We obtain that the critical point $S$ is given by $S=V^{\dg}\left(D-\frac{1}{2}\mu\imath \mathbb{I}_N\right)V$, with $V\in \UN$ coming from the spectral decomposition theorem and $D$ the diagonal matrix determined above.

Conversely, by direct computations, we can prove that the matrices $S={U}^{\dg}\left(D-\frac{1}{2}\mu\imath \mathbb{I}_N\right){U}$, with $U\in \UN$ an arbitrary unitary matrix and $D$ the diagonal matrix determined above, are critical points for the reduced cost function $\widehat{G}$.\medskip

The algorithm for solving the problem \eqref{critice-I-N} of finding the critical points for $\widehat{G}$ has the following steps:
\begin{enumerate} [(1)]
\item Run the value of $|\mathcal{K}_+|$ through the set $\{0,1,\dots,N\}$.

\item Solve the equation $z^{N-2|\mathcal{K}_+|}=(-1)^{N-|\mathcal{K}_+|}$ and select the solutions with $\text{Re}(z)\geq 0$ and $|z|=1$.

\item For every solution $z$ from step (2) take $\mu=2\text{Im}(z)$.

\item  Construct the diagonal matrix 
$$D_{|\mathcal{K}_+|,\mu}=\sqrt{1-\frac{\mu^2}{4}}\,\text{diag}(\underbrace{1,\dots,1}_{|\mathcal{K}_+|\rm\ times},\underbrace{-1,\dots,-1}_{N-|\mathcal{K}_+|\rm\ times}).$$

\item Construct the critical points of the cost function $\widehat{G}$ 
$$S_{|\mathcal{K}_+|,\mu}(U)=U^{\dg}\left(D_{|\mathcal{K}_+|,\mu}-\frac{1}{2}\mu\imath \mathbb{I}_N\right)U,$$
where $U\in \UN$ is an arbitrary unitary matrix.

\end{enumerate}

The value of the reduced cost function $\widehat{G}$ in the critical point $S_{|\mathcal{K}_+|,\mu}(U)$ is 
$$\widehat{G}\left(S_{|\mathcal{K}_+|,\mu}(U)\right)=\sqrt{1-\frac{\mu^2}{4}}\cdot (2|\mathcal{K}_+|-N).$$


\subsection{The nature of the kinematic critical points  when the target matrix $A$ is from $\SUN$}

In order to compute $\text{Hess}_{\UN}\, G$ we consider the geodesic curve $\gamma:\R\rightarrow \UN$ given by $\gamma(t)= \exp(t\Omega)U$, which passes trough $U\in \UN$ at $t=0$ with the speed $\Omega$ in the Lie algebra $\mathfrak{u}(N)$ of $\UN$ (which is formed with the skew-Hermitian matrices from $\mathcal{M}_N(\mathbb{C})$).

Since we have $G(U)=\text{Re}\tr(A^{\dg}U)=\frac{1}{2}(\tr(A^{\dg}U)+\tr(U^{\dg}A))$, we write successively
\begin{align*}
&(G\circ\gamma)(t)=\frac{1}{2}\left(\tr(A^{\dg}\exp(t\Omega)U)+\tr(U^{\dg}\exp(t\Omega^{\dg})A)\right);\\
&(G\circ\gamma)'(t)=\frac{1}{2}\left(\tr(A^{\dg}\exp(t\Omega)\Omega U)+\tr(U^{\dg}\exp(t\Omega^{\dg})\Omega^{\dg}A)\right);\\
&(G\circ\gamma)''(t)=\frac{1}{2}\left(\tr(A^{\dg}\exp(t\Omega)\Omega^2 U)+\tr(U^{\dg}\exp(t\Omega^{\dg})(\Omega^{\dg})^2A)\right).
\end{align*}
We obtain the quadratic form associated with the Hessian matrix in a point $U\in\UN$
\begin{equation*}
\text{Hess}_{\UN}\, G(U)(\Omega U,\Omega U) = \left(G\circ\gamma\right)''(0)=\frac{1}{2}\left(\tr(A^{\dg}\Omega^2 U)+\tr(U^{\dg}(\Omega^{\dg})^2A)\right)=\frac{1}{2}\tr(\Omega^2(UA^{\dg}+AU^{\dg})).
\end{equation*}
By the polarization law, we compute  
$$\text{Hess}_{\UN}\, G(U)(\Omega_1 U,\Omega_2 U)= \frac{1}{4}\tr((\Omega_1\Omega_2+\Omega_2\Omega_1)(UA^{\dg}+AU^{\dg})),\, \text{for}\, \Omega_1,\Omega_2\in \mathfrak{u}(N).$$
Applying the formula \eqref{Helmke-formula11}, we obtain the Hessian for the cost function $\widetilde{G}$ in a point $S\in \SUN$
$$\text{Hess}_{\SUN}\, \widetilde{G}(S)(\Omega_1 S,\Omega_2 S)= \frac{1}{4}\tr((\Omega_1\Omega_2+\Omega_2\Omega_1)(SA^{\dg}+AS^{\dg})),\, \text{for}\, \Omega_1,\Omega_2\in \mathfrak{su}(N),$$
where the Lie algebra $\mathfrak{su}(N)$ is the set
$$\mathfrak{su}(N)=\{\Omega\in \mathcal{M}_N(\mathbb{C})\,|\,\Omega+\Omega^{\dg}=\mathbb{O}_N,\,\tr(\Omega)=0\}.$$
The Hessian for the reduced cost function $\widehat{G}$ in a point $S\in \SUN$ is given by
\begin{equation*}
\text{Hess}_{\SUN}\, \widehat{G}(S)(\Omega_1 S,\Omega_2 S)= \frac{1}{4}\tr((\Omega_1\Omega_2+\Omega_2\Omega_1)(S+S^{\dg})),\, \text{for}\, \Omega_1,\Omega_2\in \mathfrak{su}(N)
\end{equation*}
and the quadratic form associated with the Hessian matrix in a point $S\in \SUN$ is
\begin{equation*}
\text{Hess}_{\SUN}\, \widehat{G}(S)(\Omega S,\Omega S) =\frac{1}{2}\tr(\Omega^2(S+S^{\dg})),\,\text{for}\, \Omega\in \mathfrak{su}(N).
\end{equation*}

For a kinematic critical point which has the formula $S_{|\mathcal{K}_+|,\mu}(U)=U^{\dg}\left(D_{|\mathcal{K}_+|,\mu}-\frac{1}{2}\mu\imath \mathbb{I}_N\right)U$, the quadratic form associated with the Hessian matrix in the point $S_{|\mathcal{K}_+|,\mu}(U)$ becomes
\begin{align*}
\text{Hess}_{\SUN}\, \widehat{G}(S)(\Omega S,\Omega S)  & =\tr\left(\Omega^2\cdot\frac{1}{2}(U^{\dg}(D_{|\mathcal{K}_+|,\mu}-\mu\imath\mathbb{I}_N)U+U^{\dg}(D_{|\mathcal{K}_+|,\mu}+\mu\imath\mathbb{I}_N)U)\right)\\
& = \tr\left(\Omega^2U^{\dg}D_{|\mathcal{K}_+|,\mu}U\right) = \tr\left(U\Omega U^{\dg}U\Omega U^{\dg}D_{|\mathcal{K}_+|,\mu}\right) \\
& = \tr\left(\widetilde{\Omega}^2D_{|\mathcal{K}_+|,\mu}\right),
\end{align*}
where we have denoted $\widetilde{\Omega}=U\Omega U^{\dg}\in \mathfrak{su}(N)$.

The following result presents the nature of the critical points for the reduced cost function $\widehat{G}$ (and consequently for the trace fidelity cost function $\widetilde{G}$).

\begin{thm} The exhaustive list of critical points for $\widehat{G}$ and their nature is as follows.

\begin{enumerate}[(i)]
\item For $0<|\mathcal{K}_+|<N$, the critical points $S_{|\mathcal{K}_+|,\mu}(U)$ are all saddle points. 

\item For $|\mathcal{K}_+|=0$, the critical points are $S_{0,\mu}(U)=-\left(\sqrt{1-\frac{\mu^2}{4}}+\frac{1}{2}\mu\imath \right)\mathbb{I}_N$, where $\mu=2\emph{Im}(z)$, with $z$ being the solutions of the equation $z^N=(-1)^N$ verifying $\emph{Re}(z)\geq 0$. 

\begin{enumerate}[(a)]
\item In the case $N=2P$, $P\in \mathbb{N}$, the critical point $S_{0,0}=-\mathbb{I}_N$ is the only global minimum. In addition, when $N$ is a multiple of 4 the critical points $S_{0,\pm 2}=\mp \imath\mathbb{I}_N$ are saddle points.
 All the other critical points $S_{0,\mu}$, which exist for $P\geq 3$, {\bf are local minima which are not global minima}.

\item In the case $N=2P+1$, $P\in \mathbb{N}$, the critical points $S_{0,\mu}$, where $\mu=2\text{Im}(z)$ for $z=-\cos\left(\frac{2P\pi}{2P+1}\right)\pm \sin\left(\frac{2P\pi}{2P+1}\right)\imath$, are the global minima. All the other critical points $S_{0,\mu}$, which exist for $P\geq 3$, {\bf are local minima which are not global minima}.

\end{enumerate}
\item For $|\mathcal{K}_+|=N$, the critical points are $S_{N,\mu}(U)=\left(\sqrt{1-\frac{\mu^2}{4}}-\frac{1}{2}\mu\imath \right)\mathbb{I}_N$, where $\mu=2\emph{Im}(z)$, with $z$ being the solutions of the equation $z^N=1$ verifying $\emph{Re}(z)\geq 0$. The critical point $S_{N,0}=\mathbb{I}_N$ is the only global maximum. All the other critical points, which exist for $N\geq 5$, {\bf are local maxima which are not global maxima}.

\end{enumerate}
\end{thm}

\begin{proof}

$(i)$ 
It is easy to notice that we have a critical point with $\mu=\pm 2$ if and only if the complex numbers $\imath$ or $-\imath$ are solutions for the equation $z^{N-2|\mathcal{K}_+|}=(-1)^{N-|\mathcal{K}_+|}$ (see the second step of the algorithm for determining the critical points), which is equivalent with $N$ being a multiple of 4.\medskip

{\bf I.} If $N$ is not a multiple of 4, then $\mu\neq \pm 2$.  A critical point of the cost function $\widehat{G}$ is of the form
$$S_{|\mathcal{K}_+|,\mu}(U)=U^{\dg}\left(D_{|\mathcal{K}_+|,\mu}-\frac{1}{2}\mu\imath \mathbb{I}_N\right)U,$$
where $U\in \UN$ and $$D_{|\mathcal{K}_+|,\mu}=\sqrt{1-\frac{\mu^2}{4}}\,\text{diag}(\underbrace{1,\dots,1}_{|\mathcal{K}_+|\rm\ times},\underbrace{-1,\dots,-1}_{N-|\mathcal{K}_+|\rm\ times}).$$
Note that $0<|\mathcal{K}_+|<N$. 
We have seen that 
$$\text{Hess}_{\SUN}\, \widehat{G}(S_{|\mathcal{K}_+|,\mu}(U))(\Omega S_{|\mathcal{K}_+|,\mu}(U),\Omega S_{|\mathcal{K}_+|,\mu}(U))=\sqrt{1-\frac{\mu^2}{4}}\tr\left(\widetilde{\Omega}^2\text{diag}(\underbrace{1,\dots,1}_{|\mathcal{K}_+|\rm\ times},\underbrace{-1,\dots,-1}_{N-|\mathcal{K}_+|\rm\ times})\right),$$ 
with $\widetilde{\Omega}=U\Omega U^{\dg}\in \mathfrak{su}(N)$.

We look for matrices $\widetilde{\Omega}\in \mathfrak{su}(N)$ of the form
$$\widetilde{\Omega}=\text{diag}(-(d_2+\dots+d_N)\imath,d_2\imath,\dots,d_N\imath),$$
with arbitrary $d_2,\dots,d_N\in \R$ such that $\tr(\widetilde{\Omega}^2\text{diag}(\underbrace{1,\dots,1}_{|\mathcal{K}_+|\rm\ times},\underbrace{-1,\dots,-1}_{N-|\mathcal{K}_+|\rm\ times}))>0$ and respectively $\tr(\widetilde{\Omega}^2\text{diag}(\underbrace{1,\dots,1}_{|\mathcal{K}_+|\rm\ times},\underbrace{-1,\dots,-1}_{N-|\mathcal{K}_+|\rm\ times}))<0$.

By direct computation we obtain 
\begin{align*}
\tr(\widetilde{\Omega}^2\text{diag}(\underbrace{1,\dots,1}_{|\mathcal{K}_+|\rm\ times},\underbrace{-1,\dots,-1}_{N-|\mathcal{K}_+|\rm\ times}))= & -2\left(d_2+\dots+d_{|\mathcal{K}_+|}+d_{|\mathcal{K}_+|+2}+\dots+d_N\right)d_{|\mathcal{K}_+|+1} \\
& +\text{terms depending only on}\,\, d_2,\dots,d_{|\mathcal{K}_+|},d_{|\mathcal{K}_+|+2},\dots,d_N.
\end{align*}
If we fix the values of $d_2,\dots,d_{|\mathcal{K}_+|},d_{|\mathcal{K}_+|+2},\dots,d_N$ such that they have nonzero sum, then the expression above can be regarded as a non-constant linear function of variable $d_{|\mathcal{K}_+|+1}$ , which can take any real value.
Therefore, the critical point $S_{|\mathcal{K}_+|,\mu}(U)$ is a saddle point.\medskip

{\bf II.} If $N=4M$, $M\in \mathbb{N}$, and $|\mathcal{K}_+|=2M$, the equation from the second step of the algorithm becomes an identity and, consequently, we obtain a continuous family of critical points parameterized by $\mu\in [-2,2]$ 
$$S_{2M,\mu}(U)=U^{\dg}\text{diag}\left(\underbrace{\sqrt{1-\frac{\mu^2}{4}}-\frac{\mu\imath}{2}, \dots,\sqrt{1-\frac{\mu^2}{4}}-\frac{\mu\imath}{2}}_{2M\,\text{times}},\underbrace{ -\sqrt{1-\frac{\mu^2}{4}}-\frac{\mu\imath}{2}, \dots ,-\sqrt{1-\frac{\mu^2}{4}}-\frac{\mu\imath}{2}}_{2M\,\text{times}} \right)U,$$ with an arbitrary $U\in \mathcal{U}(N)$.

The Hessian matrix is given by
$$\text{Hess}_{\SUN}\, \widehat{G}(S_{2M,\mu}(U))(\Omega S_{2M,\mu}(U),\Omega S_{2M,\mu}(U))=\sqrt{1-\frac{\mu^2}{4}}\tr\left(\widetilde{\Omega}^2\text{diag}(\underbrace{1,\dots,1}_{2M\rm\ times},\underbrace{-1,\dots,-1}_{2M\rm\ times})\right),$$ 
with $\widetilde{\Omega}=U\Omega U^{\dg}\in \mathfrak{su}(N)$.

If $\mu\neq \pm 2$, then we obtain a continuous family of saddle points $S_{2M,\mu}(U)$, using the same argument as above.

For the limit cases $\mu=\pm 2$ the Hessian matrices are identically zero, but these critical points are also saddle points (see Appendix).\medskip

$(ii)$ For the case $(a)$, which corresponds to $N=2P$, the critical points depend on the solutions of the equation $z^N=1$ for which $\text{Re}(z)\geq 0$. The solution $z=1$ of this equation leads to the critical point $S_{0,0}$. We notice that $\widehat{G}(S_{0,0})=-N$ and $\widehat{G}(S)> -N$ for any critical point $S\neq S_{0,0}$. In conclusion, $S_{0,0}=-\mathbb{I}_N$ is the only global minimum (which exists because $\SUN$ is a compact manifold). When $N$ is a multiple of 4, then $z=\pm\imath$ are also solutions of the equation $z^N=1$. Therefore $S_{0,\pm 2}=\mp\imath\mathbb{I}_N$   are critical points and since $S_{0,\pm 2}=S_{P,\pm 2}$, it follows from $(i)$ that $S_{0,\pm 2}$ are saddle points. For all the other critical points, we have 
$$\text{Hess}_{\SUN}\, \widehat{G}(S_{0,\mu})(\Omega S_{0,\mu},\Omega S_{0,\mu})= -\sqrt{1-\frac{\mu^2}{4}}\tr(\widetilde{\Omega}^2)=\sqrt{1-\frac{\mu^2}{4}}\|\widetilde{\Omega}^2\|_F.$$
Since $\widehat{G}(S_{0,\mu})>-N$, we obtain that $S_{0,\mu}$ are local minima which are not global minima.

For the case $(b)$, finding the critical points reduces to solve the equation $z^N=-1$, which has the solutions 
$z_k=-\cos\left(\frac{2k\pi}{2P+1}\right)\pm\sin\left(\frac{2k\pi}{2P+1}\right)\imath$, $k\in \{0,\dots,P\}$. When $\cos\left(\frac{2k\pi}{2P+1}\right)\leq 0$ we obtain critical points with $\mu_k=\pm 2\sin\left(\frac{2k\pi}{2P+1}\right)$ and with the same argument as above the Hessian in these critical points is positive definite. In conclusion, all these critical points are local minima. We have that $\widehat{G}(S_{0,\mu_k})=-N\sqrt{1-\frac{\mu_k^2}{4}}=N\cos\left(\frac{2k\pi}{2P+1}\right)$. We observe that $\widehat{G}(S_{0,\mu_P})=\widehat{G}(S_{0,-\mu_P})<\widehat{G}(S_{0,\mu_k})$ for all $k\neq P$. Therefore, the critical points $S_{0,\mu_P}$ and $S_{0,-\mu_P}$ are global minima and all the other critical points are local minima which are not global minima.

$(iii)$ We have $\widehat{G}(S_{N,0})=N$ and $\widehat{G}(S)< N$ for any critical point $S\neq S_{N,0}$. In conclusion, $S_{N,0}=\mathbb{I}_N$ is the only global maximum (which exists because $\SUN$ is a compact manifold). The argument that all the other critical points are local maxima which are not global maxima is analogous with the one used in $(ii)$.

\end{proof}

Regarding the quantum control landscape problem we can synthesize the above results in the following theorem.
We remind that we work under the hypothesis that the quantum controlled system \eqref{control} is globally controllable.

\begin{thm} Assume that the quantum control landscape function $J$ is constructed with the trace fidelity function $\widetilde{G}$, where the target gate $A$ is a special unitary matrix from $\SUN$. Then,
\begin{enumerate}[(i)]

\item for $N\geq 5$, the landscape function $J$ is {\bf not} $\SUN$-trap free (i.e. there are always kinematic local extrema that are not global extrema);

\item  for $N<5$, the landscape function $J$ is kinematic $\SUN$-trap free (in the sense that there are no kinematic local extrema that are not global extrema). 
\end{enumerate}

\end{thm}

\medskip


\appendix
\section{Appendix}

\begin{lem}
Let $(\mathfrak{M},{\bf g})$ be a Riemannian manifold and $f:\mathfrak{M}\rightarrow \R$ be a smooth function. Let $\lambda:[a,b]\subset\R\rightarrow \mathfrak{M}$ be a continuous curve of critical points for $f$, such that $f(\lambda(t))$ is constant for $t\in [a,b]$. Also, assume that, for every $t\in (a,b)$, the critical points $\lambda(t)$ are saddle points\footnote{A critical point $x_0\in \mathfrak{M}$ is called saddle point for $f$ if in any neighborhood of $x_0$ there exist $p,q$ such that $f(p)<f(x_0)<f(q)$.}. Then, $\lambda(a)$ and $\lambda(b)$ are also saddle points for $f$.
\end{lem}

\begin{proof}
Let ${\bf B}(\lambda(a),\varepsilon)$ the ball centered in $\lambda(a)$ having radius $\varepsilon>0$. From the continuity of $\lambda$ there exists $c\in (a,b)$ such that $\lambda(c)\in {\bf B}(\lambda(a),\varepsilon)$. We choose a radius $\eta>0$ such that ${\bf B}(\lambda(c),\eta)\subset {\bf B}(\lambda(a),\varepsilon)$. Since $\lambda(c)$ is a saddle point for $f$ there exist $p,q\in  {\bf B}(\lambda(c),\eta)$ such that $f(p)<f(\lambda(c))<f(q)$. Because $f(\lambda(a))=f(\lambda(c))$ it follows that $p,q \in {\bf B}(\lambda(c),\eta)$ and $f(p)<f(\lambda(a))<f(q)$. Therefore, $\lambda(a)$ is also a saddle point for $f$. A similar argument holds for the critical point $\lambda(b)$.
\end{proof}


\end{document}